\documentclass[runningheads]{llncs}
\usepackage{graphicx}
\usepackage{graphicx}
\usepackage[T1]{fontenc}
\pdfoutput=1
\usepackage{paralist}
\usepackage{latexsym}
\usepackage{epsfig}
\usepackage{amsfonts}
\usepackage{amsmath}
\usepackage{amssymb}
\usepackage{mathrsfs}
\usepackage{epsfig}
\usepackage{amsmath}
\usepackage{subfigure}
\usepackage{multirow}
\usepackage{array}
\usepackage[font=small,labelfont=bf,tableposition=top]{caption}
\usepackage{tabularx}
\usepackage{tikz}
\usepackage{graphicx}
\usepackage{float}
\usepackage{url}
\usepackage{multicol}
\usepackage{enumerate}
\usepackage{makecell,rotating,multirow,diagbox}
\usepackage{wrapfig}
\usepackage{subfigure}
\usepackage{footmisc}
\usepackage{amssymb}
\usepackage{graphicx}

\usepackage[bookmarksnumbered, bookmarksopen,colorlinks,citecolor=black,urlcolor=black, anchorcolor=black,linkcolor=black]{hyperref}

\usetikzlibrary{automata, positioning}
\usetikzlibrary{arrows}
\usetikzlibrary{shapes}
\usetikzlibrary{decorations}
\usetikzlibrary{calc}
\usetikzlibrary{patterns,snakes}
\usetikzlibrary{intersections}
\tikzstyle{state}=[draw,circle,inner sep=0pt,minimum size=6mm,thick]
\tikzstyle{start}=[pin={[pin edge={black,stealth'-}]left:}]
\tikzstyle{start above}=[pin={[pin edge={black,<-}]above:}]
\tikzstyle{accepting}=[double distance=0.5pt]
\tikzstyle{transition}+=[inner sep=1pt]
\tikzstyle{htransition}=[transition,minimum width=1.5em]
\tikzstyle{vtransition}=[transition,minimum height=1.5em]

\usetikzlibrary{automata}
\usetikzlibrary{er}
\tikzstyle{every initial by arrow}=[initial text=]

\newcommand{\PreserveBackslash}[1]{\let\temp=\\#1\let\\=\temp}
\begin{document}
\title{Discovering an Algorithm Actually Learning Restricted Single Occurrence Regular Expression with Interleaving}
%
\titlerunning{Abbreviated paper title}


\author{Xiaofan Wang \and Xiaolan Zhang}
\authorrunning{X.Wang et al.}
%
\institute{State Key Laboratory of Computer Science, Institute of Software, Chinese Academy of Sciences \\ University of Chinese Academy of Sciences \\
\email{\{wangxf,zhangxl\}@ios.ac.cn}}
%
\maketitle              
\begin{abstract}
A recent paper \cite{pro1,pro2} proposed an algorithm $i$SOIRE, which
combines \textit{single-occurrence automaton} (SOA) \cite{Geert2010,Frey2015} and \textit{maximum independent set} (MIS) to learn a
subclass \textit{single-occurrence regular expressions with interleaving} (SOIREs) and claims the learnt expression is SOIRE, which has unrestricted usage for interleaving.
However, in reality, the learnt expression still has many restrictions for using interleaving, even does for Kleene-star or interation, i.e,  the learnt expression is not an SOIRE,
we prove that by examples.  
In this paper, for the algorithm $i$SOIRE, we first give the basic notions,
then provide analyses about incorrectness, finally present the correct result learnt by $i$SOIRE.
Our theoretical analyses demonstrate that the result derived by $i$SOIRE  belongs to a subclass of SOIREs.
\end{abstract}

%
%
\section{Basic Notions}

We give the notions about single-occurrence regular expressions with interleaving (SOIRE), single-occurrence automaton (SOA).
\begin{definition}[regular expression with interleaving]
Let $\Sigma$ denote a finite set of alphabet symbols.
The regular expression with interleaving are defined as follows.
$\varepsilon$, $a\in\Sigma$ are regular expressions.
For the regular expressions $r_1$ and $r_2$, the concatenation $r_1\cdot r_2$, the kleene-star $r_1^*$, the disjunction $r_1|r_2$,
the interleaving $r_1\&r_2$ are also regular expressions. Note that iteration $r_1^{+}$, optional $r?$ are used as abbreviations of $r_1r_1^*$, $r|\varepsilon$.
Usually, we omit concatenation operators in examples.
$\mathcal{L}(r_1\&r_2)\!=\!\mathcal{L}(r_1)\&\mathcal{L}(r_2)=\bigcup_{s_1\in \mathcal{L}(r_1),s_2\in \mathcal{L}(r_2)} s_1\&s_2$.
For $u,v\in \Sigma^*$ and $a,b\in \Sigma$, $u\&\varepsilon=\varepsilon \&u=\{u\}$,
 and $(au)\&(bv)=\{a(u\&bv)\}\cup \{b(au\& v)\}$.
\end{definition}

\begin{definition}[single-occurrence regular expression with interleaving (SOIRE)]
Let $\Sigma$ be a finite alphabet. A single-occurrence regular expression with interleaving
 (SOIRE) is a regular expression with interleaving over $\Sigma$ in which every terminal
symbol occurs at most once.
\end{definition}

According to the definition, SOIREs have unrestricted usage for interleaving and other operators.
\begin{definition}[single-occurrence automaton (SOA) \cite{Geert2010,Frey2015}]
Let $\Sigma$ be a finite alphabet, and let $q_0$, $q_f$ be distinct symbols
that do not occur in $\Sigma$. A single-occurrence automaton (SOA) over $\Sigma$ is a finite
directed graph $\mathscr{A}\! =\! (V ,E)$ such that
(1) $\{q_0, q_f\}\in V$, and $V \subseteq \Sigma \cup \{q_0, q_f\}$.
(2) $q_0$ has only outgoing edges, $q_f$ has only incoming edges, and every $v\in V\setminus \{q_0,q_f\}$ is visited
during a walk from  $q_0$ to $q_f$.
\end{definition}

A string $a_1\cdots a_n$ ($n\geq 0$) is accepted by an SOA $\mathscr{A}$, if and only if there is a path $q_0\rightarrow a_1\rightarrow\cdots\rightarrow a_n\rightarrow q_f$ in $\mathscr{A}$.

\section{Analyses about the algorithm $i$SOIRE}
\label{analysis}

First, we present analyses about the algorithm $i$SOIRE.
Then, in term of the expression learnt by the algorithm $i$SOIRE, we obtain two conclusions and provide the corresponding proofs.
 The two conclusions reveal the expression learnt by the algorithm $i$SOIRE is not an SOIRE.

 For any given finite sample, SOA is constructed by using algorithm 2T-INF \cite{Geert2010},
 the algorithm $i$SOIRE \cite{pro1,pro2} combines SOA and MIS to infer an expression called SOIRE.
The main procedure $Soa2Soire$ \cite{pro1,pro2}  (see Figure \ref{Alg1}) is designed by revising only few steps of the algorithm $Soa2Sore$ \cite{Frey2015} (see Figure \ref{AlgSoa2Sore}), which is used to infer a single-occurrence regular expression (SORE) \cite{Geert2010,Frey2015}.
The main differences between algorithms $Soa2Soire$ and $Soa2Sore$ are in line 5 $\sim$ line 8.

In algorithm $Soa2Sore$, in line 5 $\sim$ line 8,
if the input SOA built form given finite sample contains a strongly connected component ($U$),
$U$ is used to infer an expression ($r$) and then is added an iteration operator ($^+$), i.e., $r^+$.
The strongly connected component ($U$) is contracted to a vertex labelled with $r^+$.
However, in algorithm $Soa2Soire$, 
only for the strongly connected component ($U$) where $|U|\!=\!1$,
$U$ is used to generate an expression ($a\!\in\! \Sigma,U\!=\!\{a\}$) and then is added an iteration operator ($^+$), i.e., $a^+$.
For $|U|\!>\!1$, $U$ is used to introduce interleaving $\&$ into inferred expression by calling subroutine $Merge$ (see Figure \ref{Alg2}).

Subroutine $Merge$ is used to return an expression, where the interleaving $\&$ is introduced.
$filter(U,S)$ \cite{pro1,pro2}  denotes that, for each string $s\in S$,
$filter$ extracts the substring consisting of symbols in $U$, where the order of the alphabetic symbols maintains the relative order of that in $s$.
The input of the subroutine $Merge$ is the set of the strings extracted by $filter$.
In $Merge$, first, in line 1 $\sim$ line 8, $all\_mis$ \cite{peng2015} (the set of the distinct maximum independent set) is computed.
Then, in line 9 $\sim$ line 13, for each obtained maximum independent set $mis\in all\_mis$, $mis$ is used to generate an expression by recursively calling $Soa2Soire$,
where the input is the set of the strings extracted by $filter(mis,S)$.
Each generated expression is input in $U'$.
Finally, in line 14, subroutine $combine$ connects all expressions in $U'$ by using interleaving $\&$.
For example $U'\!=\!\{r_1,r_2,r_3\}$, $combine(U')\!=\!r_1\&r_2\&r_3$.
The result returned by $Merge$ is the result of $combine$.

Since the algorithm $Soa2Soire$ presented in \cite{pro1,pro2} does not provide the proof about the correctness,
for any learnt expression, we check whether the learnt expression is SOIRE, which has unrestricted usage for interleaving and other operators.
However, we discover that the learnt expression still has many restrictions for using interleaving, even does for Kleene-Star or iteration,
the learnt result is not an SOIRE.
For better understanding, we provide the proofs by examples, which specify the learnt result is not an SOIRE.
\begin{figure}[H]
  \centering
  \includegraphics[width=0.75\textwidth]{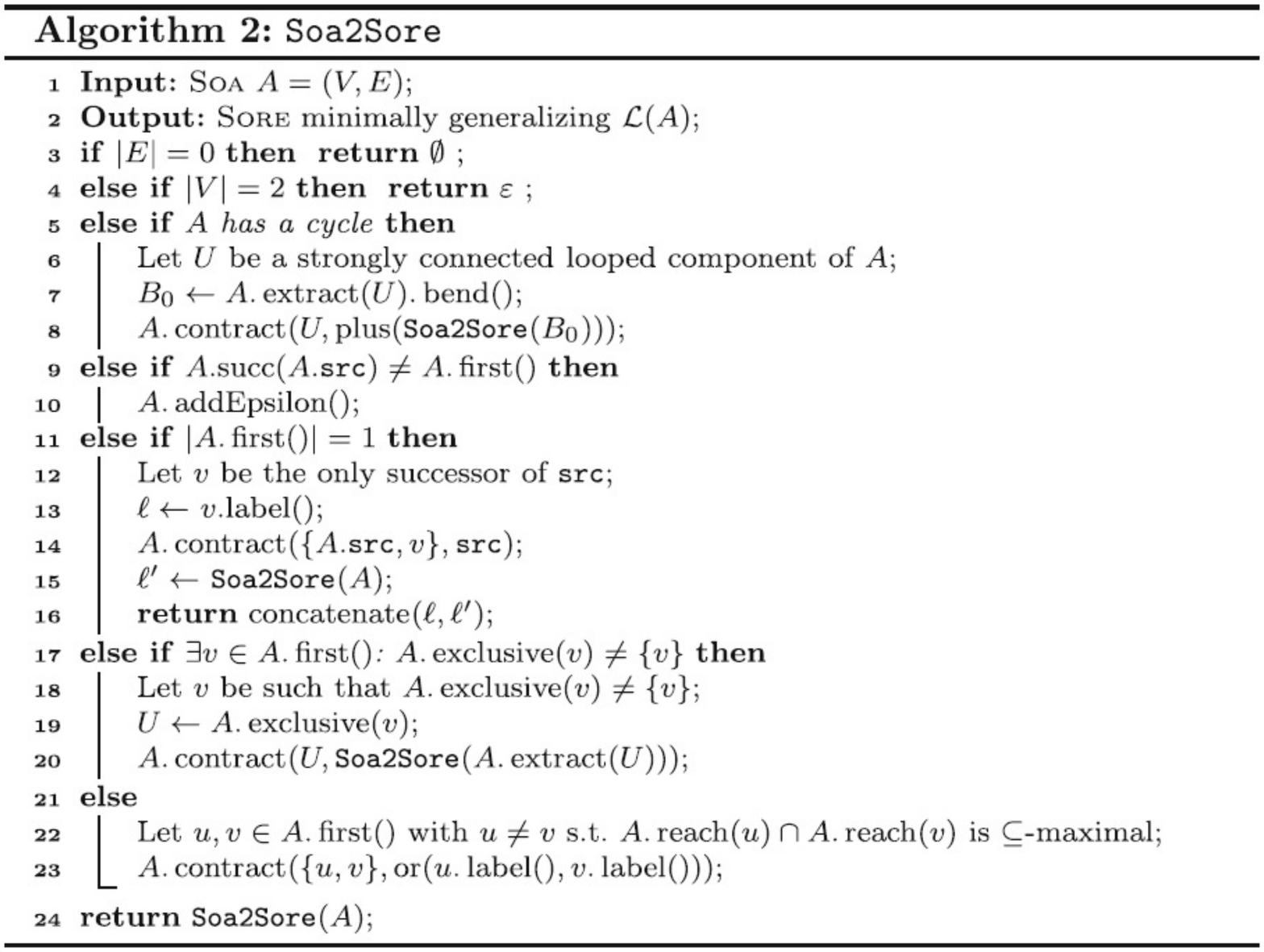}
  \caption{The algorithm $Soa2Sore$.}
  \label{AlgSoa2Sore}
\end{figure}
\vspace{-1.5cm}
\begin{figure}[H]
  \centering
  \subfigure[$Soa2Soire$]{
  \includegraphics[width=0.45\textwidth]{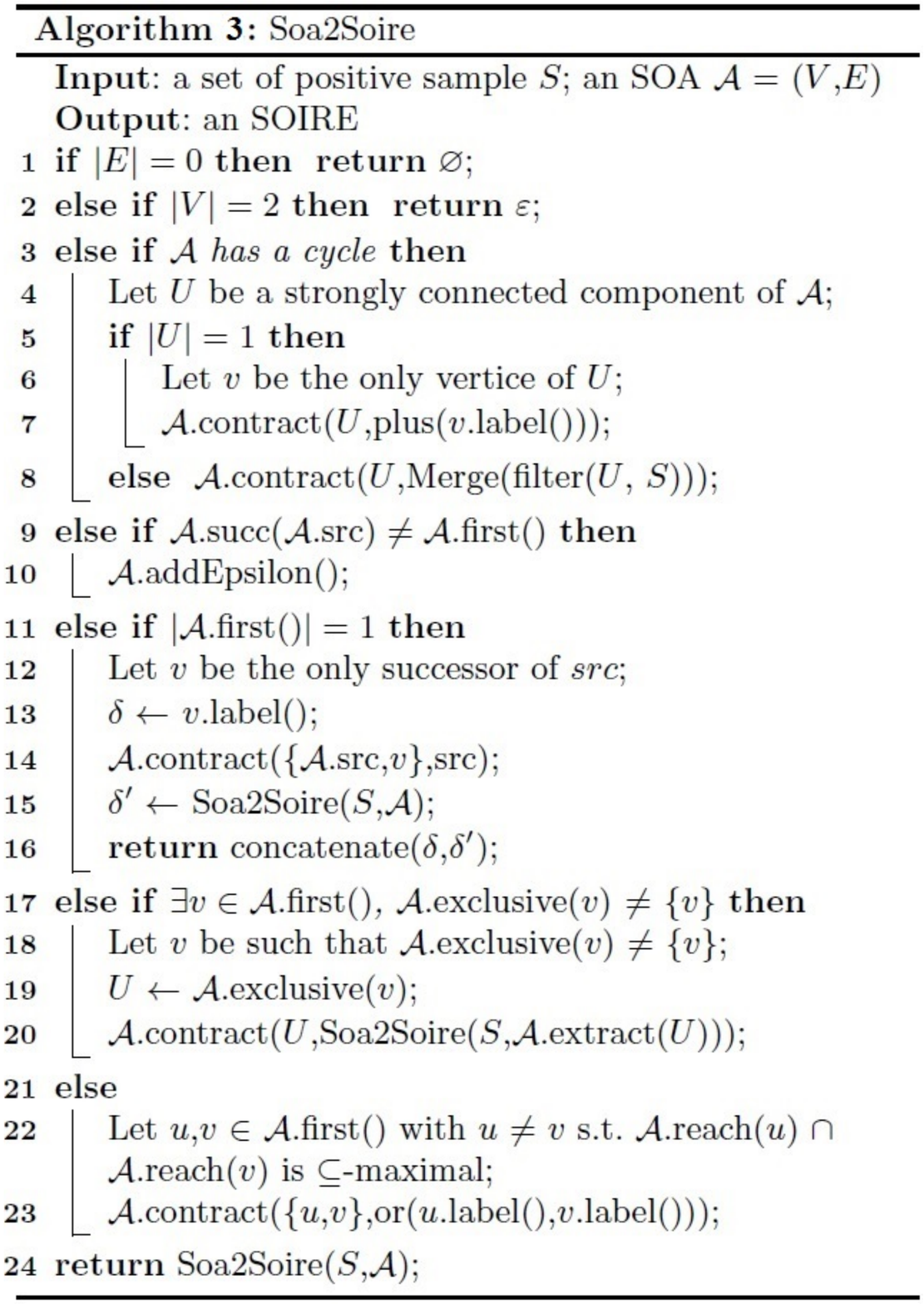}
  \label{Alg1}
  }
  \subfigure[$Merge$]{
  \includegraphics[width=0.45\textwidth]{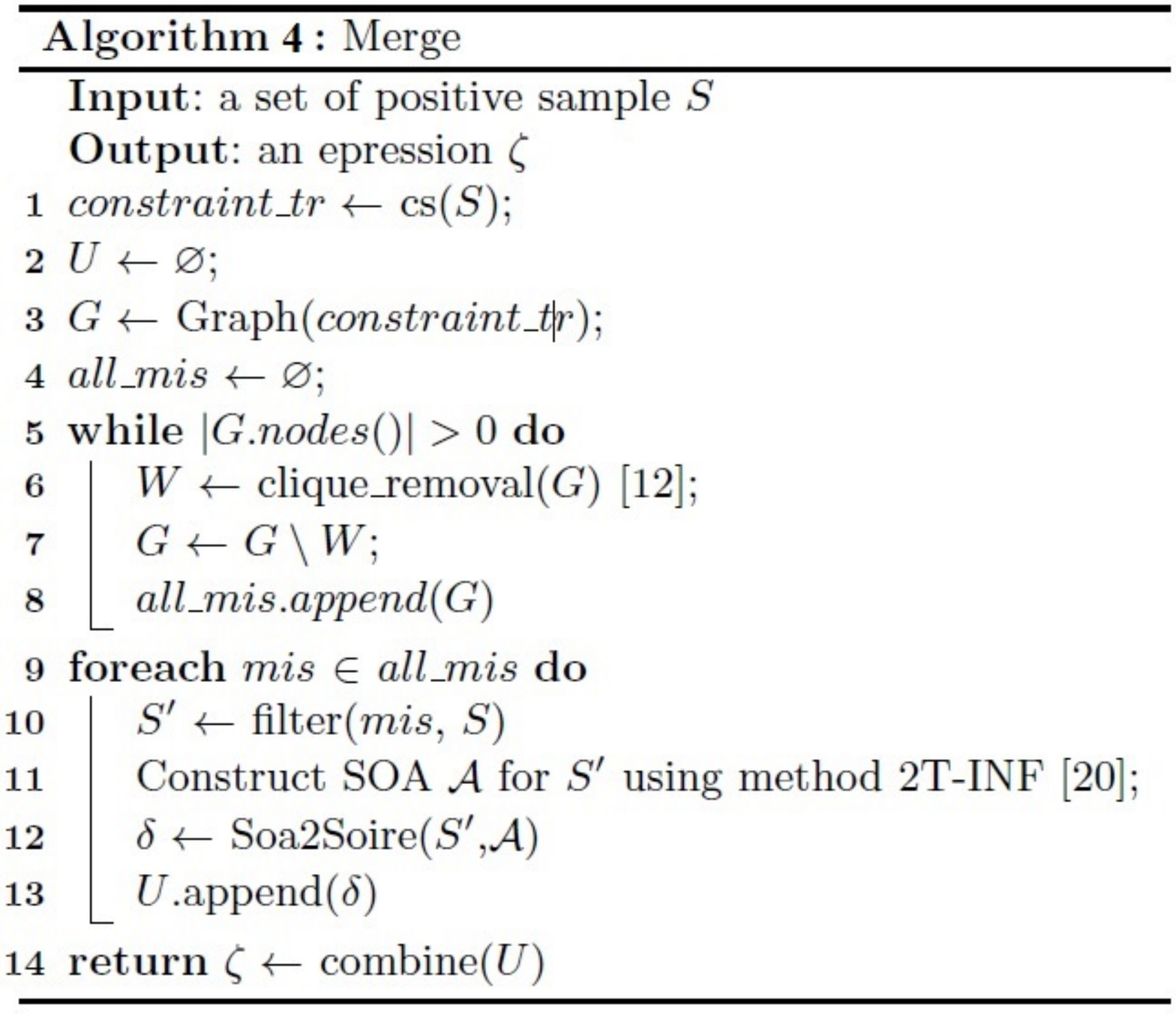}
  \label{Alg2}
  }
  \caption{The algorithm $i$SOIRE.}
  \label{AlgiSOIRE}
\end{figure}

For the given algorithm $Soa2Soire$, there are two conclusions:
\begin{enumerate}
\item  For the learnt expression, Kleene-star $^*$ or iteration $^+$ is just on a single alphabet symbol.
\begin{proof}
  In algorithm $Soa2Soire$, in line 6 $\sim$ line 7, $plus()$ \cite{Frey2015},
  which should have used to add an iteration $^+$ on an expression, but only works on a single alphabet symbol.
  In the whole process of the recursion of the algorithm, $plus()$ does not occur in other cases.
  Since the initially obtained strongly connected components have been merge into a vertex (see subroutine $Contract$ \cite{Frey2015}),
  and the corresponding edges, which can form loops in SOA, have been removed (see subroutine $bend$ \cite{Frey2015}),
  even though for a node $v$: $v.label()$ is an expression consisting of multiple alphabet symbols, $v$ does not possible occur in a strongly connected component.
  Hence, we can assert that a class of expressions cannot be learnt, such as $(ab)^+$, $(a|b)^+$, $(a|b\&c)^+$, $(a\&b)^+$ and $((a|b)\&(c|d)\&(e|f))^*$.
\end{proof}
\item  The expressions of form $a\&(b(c\&d))$ cannot be learnt.
\begin{proof}
  Assume that the expression $a\&(b(c\&d))$ can be correctly learnt by induction.
  In $merge$, the interleaving $\&$ is introduced by the subroutine $combine$,   in line 9 $\sim$ line 13,
  any obtained maximum independent set $mis\in all\_mis$ can form a sample $S'$ $(S'\!=\!filter(mis,S))$.
  And SOA $A$=2T-INF$(S')$. Sample $S'$ and SOA $A$ are as inputs of the algorithm $Soa2Soire$, let $delta$ denote the derived expression.
  According to the computation of $filter$, the alphabet symbols in $\delta$ are in the same maximum independent set $mis$.

  Then in line 9, for each $mis_i\!\in\! all\_mis$, the corresponding expression $\delta_i$ is generated.
  ($S'_i\!=\!filter(mis_i,S)$; SOA $A_i$=2T-INF$(S'_i)$; $\delta_i\!=\!Soa2Soire(S'_i,A_i)$;)
  $\delta_i$ is put into $U$, then $combine(U)=\delta_1\&\delta_2\&\cdots\&_i$.

  For expression $a\&(b(c\&d))$, $mis_1\!=\!\{a\}$, $mis_2\!=\!\{b,c,d\}$, $mis_3\!=\!\{c\}$ and $mis_4\!=\!\{d\}$.
  Then for the initially constructed SOA,
  since $b,c,d$ can form a maximum independent set $mis_2$, then for $mis_3\!=\!\{c\}$ and $mis_4\!=\!\{d\}$ are not maximum independent sets, respectively.
  There is a contradiction, the initial assumption does not hold.
  The expressions of form $a\&(b(c\&d))$ cannot be learnt by $Soa2Soire$.
  \end{proof}
\end{enumerate}

Conclusion 1 and conclusion 2 have revealed that the expression learnt by $Soa2Soire$ still has many restrictions for using interleaving.
Such as the expressions of form  $(a|b\&c)^+$, $(a\&b)^+$, $((a|b)\&(c|d)\&(e|f))^*$ and $a\&(b(c\&d))$. Meanwhile,
the learnt expression  has restrictions for using Kleene-star or iteration, even  does for the expressions of form  $(ab)^+$ and $(a|b)^+$.
Actually, for conclusion 2, if $a$, $b$, $c$ and $d$ are replaced with regular expressions, respectively.
The same conclusions can be obtained. This implies that the expression learnt by $Soa2Soire$  is not an SOIRE.

\section{The correct result learnt by $i$SOIRE}

In Section \ref{analysis}, we have proved that the expression learnt by $Soa2Soire$  is not an SOIRE.
However, we should check whether the expression returned by $i$SOIRE belongs to a subclass of SOIREs or not.
Here, by analyzing the algorithm $Soa2Soire$, we propose a subclass of SOIREs called RSOIREs (see definition \ref{defRSOIRE}) and prove
that the expression learnt by $Soa2Soire$ is a RSOIRE.

\begin{definition}[restricted SOIREs]
A restricted SOIRE (RSOIRE) is  a regular expression with interleaving over $\Sigma$ by the following grammar,
and where every terminal symbol occurs at most once.
\begin{equation}
P:=SP \Big|PS\Big|S\Big|T\Big|P|S
\label{gramm1}
\end{equation}
\vspace{-0.5cm}
\begin{equation}
S:=S\&S\Big|T
\label{gramm2}
\end{equation}
\vspace{-0.5cm}
\begin{equation}
T:=T|T\Big|TT\Big|\varepsilon\Big|a\Big|a^* (a\in \Sigma)
\label{gramm3}
\end{equation}
\label{defRSOIRE}
\end{definition}
\begin{example}
$(a^+|b)(c\&d)$, $ad\&(b|c^*)$, and $a^+|b^+\&c^*$ are RSOIREs.
However, $(ab)\&(c|d)^+$, $((a|b\&c)d?)^*$, and $a\&(b(c\&d))$ are SOIREs, not RSOIREs.
\end{example}
\begin{theorem}
For any given finite sample $S$, let SOA $A$=2T-INF$(S)$, and $r=Soa2Soire(S,A)$. Then $r$ is a RSOIRE, and for any RSOIRE $r'$,
$r'$ can be learnt by $Soa2Soire$.
\label{thm1}
\end{theorem}
\begin{proof}
\begin{enumerate}
\item $r$ is a RSOIRE.\label{item1}

(1) For regular expressions $\varepsilon,a,a^*$,
    in algorithm $Soa2Soire$, in line 2, line 11 $\sim$ line 16 and  line 7, they can be derived, the correctness can be ensured by
    the corresponding correctness of algorithm $Soa2Sore$.

(2) For regular expressions $r_1r_2$ and $r_1|r_2$,
    assume that $r_1$ and $r_2$ can be correctly derived by $Soa2Soire$ by induction.
    In  line 11 $\sim$ line 16 and line 22 $\sim$ line 23, $r_1r_2$ and $r_1|r_2$ can be derived, the correctness can also be ensured by
    the corresponding correctness of algorithm $Soa2Sore$.

(3) For regular expression $r_1\&r_2\&\cdots \&r_k$ $(k\geq 2)$, assume that $r_i$ $(1\leq i\leq k)$ can be correctly derived by $Soa2Soire$ by induction.
    According to the conclusion 2 and the corresponding proof in Section \ref{analysis}, $r_i$ cannot be possible to contain the interleaving $\&$.
    The expression $r_1\&r_2\&\cdots \&r_k$ is computed by $combine$ in $Merge$,
    $r_i$ is derived by computing the corresponding maximum independent set $mis_i$.
    And for any two distinct maximum independent sets $mis_i$ and $mis_j$ $(i\!\neq\! j)$, the symbols in $mis_i$ and the symbols in $mis_j$ can be interleaved.
    Thus, the expression $r_1\&r_2\&\cdots \&r_k$ can be correctly derived by $Soa2Soire$.

The expressions discussing in (1), (2) and (3), which are connected by using concatenation or disjunction, can form complexity expressions,
and certainly they can be decomposed into the above discussed  basic expressions.
The grammar (\ref{gramm3}) and grammar (\ref{gramm2}) presented in definition \ref{defRSOIRE} can generated the expressions discussing in (1) and (2), respectively.
And the complexity expressions formed by the expressions discussing in (1), (2) and (3) can be generated by the grammar (\ref{gramm1}).
This implies that any expression $r$ learnt by $Soa2Soire$ is a RSOIRE.  

\item  For any RSOIRE $r'$, $r'$ can be learnt by $Soa2Soire$.\label{item2}

(1) For grammar (\ref{gramm3}), the corresponding generated regular expressions can be derived by $Soa2Soire$,
the correctness can be ensured by the corresponding correctness of algorithm $Soa2Sore$.

(2) For grammar (\ref{gramm2}), according to the proofs in \ref{item1},
the corresponding generated regular expressions with interleaving can also be derived by $Soa2Soire$.

(3) For grammar (\ref{gramm1}), the generated complexity expressions can be decomposed into
the expressions produced by grammar (\ref{gramm2}) or grammar (\ref{gramm3}), then the complexity expressions can also be derived by $Soa2Soire$.

This implies that, for an expression $r'$  generated by the defined grammars, $r'$ can be learnt by $Soa2Soire$.
\end{enumerate}
\end{proof}

We give a correct class of expression that can be learnt by $Soa2Soire$, and present the corresponding proofs of correctness.
 Theorem \ref{thm1} demonstrate that the expression learnt by $Soa2Soire$ belongs to a subclass of SOIREs.

 \section{Conclusion}

 In this paper, we mainly provide analyses about the incorrectness about algorithm $i$SOIRE,
 and then present the correct a class of expressions can be learnt by algorithm $i$SOIRE,
 the corresponding proofs illustrate that the learnt expression belongs to a subclass of SOIREs.
 Since the algorithm $i$SOIRE can be used to learn other classes of expressions, such as
 \textit{$k$-occurrence regular expression with interleaving},
 the corresponding correctness depends on the correctness of the algorithm $i$SOIRE.
 The comments in this paper can be provided as a reference.

%
%
%
\bibliographystyle{splncs04}
\bibliography{Discovering an Algorithm Actually Learning Restricted Single Occurrence Regular Expression with Interleaving.bbl}
\end{document}